\documentclass[11pt]{article}
\usepackage{epsf}
\usepackage{amsmath}
\usepackage{epsfig}
\usepackage{times}
\usepackage{amssymb}
\usepackage{amsthm}
\usepackage{setspace}
\usepackage{cite}

\usepackage{algorithmic}  
\usepackage{algorithm}

\def\y{{\bf y}}

\def\x{{\bf x}}

\def\x{{\mathbf x}}

\def\x{{\bf x}}
\def\y{{\bf y}}

\def\be{\begin{equation}}
\def\ee{\end{equation}}
\def\ba{\left[\begin{array}}
\def\ea{\end{array}\right]}

\def\inter{{C_{int}}}
\def\exter{{C_{ext}}}
\def\com{{C_{com}}}

\def\psiint{{\Psi_{int}}}
\def\psiext{{\Psi_{ext}}}
\def\psicom{{\Psi_{com}}}
\def\psinet{{\Psi_{net}}}

\def\internon{{C_{int}^{+}}}
\def\externon{{C_{ext}^{+}}}
\def\comnon{{C_{com}^{+}}}

\def\psiintnon{{\Psi_{int}^{+}}}
\def\psiextnon{{\Psi_{ext}^{+}}}
\def\psicomnon{{\Psi_{com}^{+}}}
\def\psinetnon{{\Psi_{net}^{+}}}

\def\x{{\bf x}}
\def\y{{\bf y}}

\def\1{{\bf 1}}

\def\0{{\bf 0}}

\def\erf{\mbox{erf}}
\def\erfinv{\mbox{erfinv}}

\newtheorem{theorem}{Theorem}

\begin{document}

\begin{singlespace}

\title {A rigorous geometry-probability equivalence in characterization of $\ell_1$-optimization
}
\author{
\textsc{Mihailo Stojnic}
\\
\\
{School of Industrial Engineering}\\
{Purdue University, West Lafayette, IN 47907} \\
{e-mail: {\tt mstojnic@purdue.edu}} }
\date{}
\maketitle

\centerline{{\bf Abstract}} \vspace*{0.1in}

In this paper we consider under-determined systems of linear equations that have sparse solutions. This subject attracted enormous amount of interest in recent years primarily due to influential works \cite{CRT,DonohoPol}. In a statistical context it was rigorously established for the first time in \cite{CRT,DonohoPol} that if the number of equations is smaller than but still linearly proportional to the number of unknowns then a sparse vector of sparsity also linearly proportional to the number of unknowns can be recovered through a polynomial $\ell_1$-optimization algorithm (of course, this assuming that such a sparse solution vector exists). Moreover,
the geometric approach of \cite{DonohoPol} produced the exact values for the proportionalities in question. In our recent work \cite{StojnicCSetam09} we introduced an alternative statistical approach that produced attainable values of the proportionalities. Those happened to be in an excellent numerical agreement with the ones of \cite{DonohoPol}. In this paper we give a rigorous analytical confirmation that the results of \cite{StojnicCSetam09} indeed match those from \cite{DonohoPol}.

\vspace*{0.25in} \noindent {\bf Index Terms: Linear systems; Neighborly polytopes;
$\ell_1$-optimization} .

\end{singlespace}

\section{Introduction}
\label{sec:back}

The main concern of this paper is an analytical study of under-determined systems of linear equations that have sparse solutions. To that end, let us assume that there is a $k$-sparse $n$ dimensional vector $\x$ such that
\begin{equation}
\y=A\x \label{eq:system}
\end{equation}
for an $m\times n$ ($m<n$) statistical matrix $A$ and an $m\times 1$ vector $\y$ (see Figure
\ref{fig:model}; here and in the rest of the paper, under $k$-sparse vector we assume a vector that has at most $k$ nonzero
components; also, in the rest of the paper we will assume the
so-called \emph{linear} regime, i.e. we will assume that $k=\beta n$
and that the number of the equations is $m=\alpha n$ where
$\alpha$ and $\beta$ are constants independent of $n$ (more
on the non-linear regime, i.e. on the regime when $m$ is larger than
linearly proportional to $k$ can be found in e.g.
\cite{CoMu05,GiStTrVe06,GiStTrVe07}). We then look at the inverse problem: given the $A$ and $\y$ from (\ref{eq:system}) can one then recover the $k$-sparse $\x$ in (\ref{eq:system}).
\begin{figure}[htb]
\centering
\centerline{\epsfig{figure=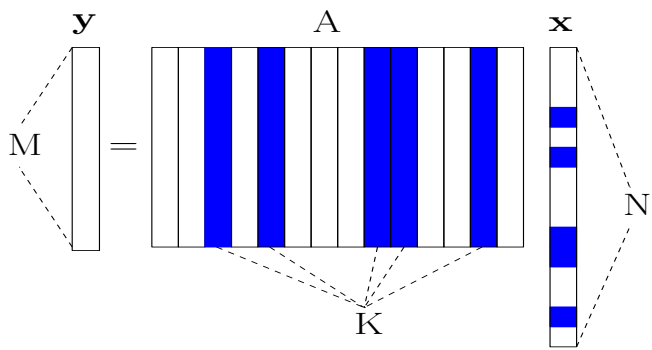,width=9cm,height=4.5cm}}
\caption{Model of a linear system; vector $\x$ is $k$-sparse}
\label{fig:model}
\end{figure}

There are of course many ways how one can attempt to recover the $k$-sparse $\x$. If one has the freedom to design $A$ in parallel with designing the recovery algorithm then the results from \cite{FHicassp,Tarokh,MaVe05} demonstrated that the techniques from
coding theory (based on the coding/decoding of Reed-Solomon codes)
can be employed to determine \emph{any} $k$-sparse $\x$ in
(\ref{eq:system}) for any $0<\alpha\leq 1$ and any
$\beta\leq\frac{\alpha}{2}$ in polynomial time. It is relatively easy to show that under the unique recoverability assumption
$\beta$ can not be greater than $\frac{\alpha}{2}$. Therefore, as long as one is concerned with the unique recovery of
$k$-sparse $\x$ in (\ref{eq:system}) in polynomial time the results from \cite{FHicassp,Tarokh,MaVe05} are
optimal. The complexity of algorithms from
\cite{FHicassp,Tarokh,MaVe05} is roughly $O(n^3)$. In a similar fashion one can, instead of using coding/decoding techniques associated with Reed/Solomon codes,
design matrix $A$ and the corresponding recovery algorithm based on the techniques related to the coding/decoding of
Expander codes (see e.g.
\cite{XHexpander,JXHC08,InRu08} and references therein). In that case recovering $\x$ in
(\ref{eq:system}) is significantly faster for large dimensions $n$. Namely, the complexity of the techniques from e.g. \cite{XHexpander,JXHC08,InRu08}
(or their slight modifications) is usually
$O(n)$ which is clearly for large $n$ significantly smaller than $O(n^3)$. However,
the techniques based on coding/decoding of Expander codes usually do not allow for $\beta$ to be as large as
$\frac{\alpha}{2}$.

If one has no freedom in the choice of the matrix $A$ (instead the matrix $A$ is rather given to us) then the recovery
problem (\ref{eq:system}) becomes NP-hard. The following two algorithms (and their different
variations) are then of special interest (and certainly have been the subject of an extensive research in recent years):
\begin{enumerate}
\item \underline{\emph{Orthogonal matching pursuit - OMP}}
\item \underline{\emph{Basis pursuit -
$\ell_1$-optimization.}}
\end{enumerate}
Under certain probabilistic assumptions on the elements of $A$ it can be shown (see e.g. \cite{JATGomp,JAT,NeVe07})
that if $m=O(k\log(n))$
OMP (or slightly modified OMP) can recover $\x$ in (\ref{eq:system})
with complexity of recovery $O(n^2)$. On the other hand a stage-wise
OMP from \cite{DTDSomp} recovers $\x$ in (\ref{eq:system}) with
complexity of recovery $O(n \log n)$. Somewhere in between OMP and BP are recent improvements CoSAMP (see e.g. \cite{NT08}) and Subspace pursuit (see e.g. \cite{DaiMil08}), which guarantee (assuming the linear regime) that the $k$-sparse $\x$ in (\ref{eq:system}) can be recovered in polynomial time with $m=O(k)$ equations.

In this paper we will focus on the second of the two above mentioned algorithms, i.e. we will focus on the performance of $\ell_1$-optimization. (Variations of the standard $\ell_1$-optimization from e.g.
\cite{CWBreweighted,SChretien08,SaZh08}) as well as those from \cite{SCY08,FL08,GN03,GN04,GN07,DG08} related to $\ell_q$-optimization, $0<q<1$
are possible as well.) Basic $\ell_1$-optimization algorithm finds $\x$ in
(\ref{eq:system}) by solving the following $\ell_1$-norm minimization problem
\begin{eqnarray}
\mbox{min} & & \|\x\|_{1}\nonumber \\
\mbox{subject to} & & A\x=\y. \label{eq:l1}
\end{eqnarray}

In seminal work \cite{CRT}, it was established that for any constant $\alpha\in (0,1)$ and $m=\alpha n$ there is a constant $\beta\in (0,\alpha)$ and $k=\beta n$ such that the solution of (\ref{eq:l1}) is with overwhelming probability the $k$-sparse $\x$ in (\ref{eq:system}) (moreover, this remains true for \emph{any} $k$-sparse $\x$). (Under overwhelming probability we in this paper assume
a probability that is no more than a number exponentially decaying in $n$ away from $1$.) The results of \cite{CRT} rested on having  matrix $A$ satisfy the restricted isometry property (RIP) which is only a \emph{sufficient}
condition for $\ell_1$-optimization to produce the solution of
(\ref{eq:system}) (more on RIP and its importance can be found in e.g. \cite{Crip,CT,Bar,Ver,ALPTJ09}).

Instead of characterizing the $m\times n$ matrix $A$ through the RIP
condition, in \cite{DonohoUnsigned,DonohoPol} Donoho associates
certain polytope with the matrix $A$. Namely,
\cite{DonohoUnsigned,DonohoPol} consider polytope obtained by
projecting the regular $n$-dimensional cross-polytope by $A$. It turns out that a \emph{necessary and sufficient}
condition for (\ref{eq:l1}) to produce the $k$-sparse solution of
(\ref{eq:system}) is that this polytope associated with the matrix
$A$ is $k$-neighborly
\cite{DonohoUnsigned,DonohoPol,DonohoSigned,DT}. Using the results
of \cite{PMM,AS,BorockyHenk,Ruben,VS,Santalo}, it is further shown in
\cite{DonohoPol}, that if $A$ is a random $m\times n$
ortho-projector matrix then with overwhelming probability polytope
obtained projecting the standard $n$-dimensional cross-polytope by
$A$ is $k$-neighborly. The precise relation between $m$ and $k$ in
order for this to happen is characterized in
\cite{DonohoPol,DonohoUnsigned} as well.

It should be noted that one usually considers success of
(\ref{eq:l1}) in recovering \emph{any} given $k$-sparse $\x$ in (\ref{eq:system}). It is also of interest to consider success of
(\ref{eq:l1}) in recovering
\emph{almost any} given $\x$ in (\ref{eq:system}). We below make a distinction between these
cases and recall on some of the definitions from
\cite{DonohoPol,DT,DTciss,DTjams2010,StojnicCSetam09,StojnicICASSP09}.

Clearly, for any given constant $\alpha\leq 1$ there is a maximum
allowable value of $\beta$ such that for \emph{any} given $k$-sparse $\x$ in (\ref{eq:system}) the solution of (\ref{eq:l1})
is exactly that given $k$-sparse $\x$  with overwhelming probability. We will refer to this maximum allowable value of
$\beta$ as the \emph{strong threshold} (see
\cite{DonohoPol}). Similarly, for any given constant
$\alpha\leq 1$ and \emph{any} given $\x$ with a given fixed location of non-zero components and a given fixed combination of its elements signs
there will be a maximum allowable value of $\beta$ such that
(\ref{eq:l1}) finds that given $\x$ in (\ref{eq:system}) with overwhelming
probability. We will refer to this maximum allowable value of
$\beta$ as the \emph{weak threshold} and will denote it by $\beta_{w}$ (see, e.g. \cite{StojnicICASSP09,StojnicCSetam09}). In this paper we will provide a rigorous proof that $\beta_w$ one can determine through Donoho's framework from \cite{DonohoPol} is exactly the same as $\beta_w$ determined in \cite{StojnicCSetam09}.

We organize the rest of the paper in the following way. In Section
\ref{sec:unsigned} we will first recall on the basic ingredients of the analysis done in \cite{DonohoPol}. Using the insights from \cite{StojnicCSetam09} we will then give a closed formula for $\beta_w$ computed in \cite{DonohoPol}. As hinted above, this formula will match the one computed in \cite{StojnicCSetam09}.
In Section \ref{sec:signed} we will then specialize the results from Section \ref{sec:unsigned} to the case when the nonzero components of sparse vector $\x$ in (\ref{eq:system}) are positive (or in general with \emph{a priori} known signs). Using again the insights from \cite{StojnicCSetam09} we will then give a closed formula for $\beta_w$ computed for this case in \cite{DT}. This formula will match the corresponding one computed in \cite{StojnicCSetam09}. Finally, in Section \ref{sec:discuss} we discuss obtained results.

\section{General $\x$}
\label{sec:unsigned}

\subsection{Success of $\ell_1$ and neighborliness of projected cross-polytope}
\label{sec:l1neigh}

In this section we show that the weak thresholds obtained in \cite{StojnicCSetam09} are the same as the ones obtained in \cite{DonohoPol}. To that end, we start by recalling on the basics of the analysis from \cite{DonohoUnsigned,DonohoPol}. In his, now legendary, paper \cite{DonohoUnsigned} Donoho took a geometric approach to the performance analysis of $\ell_1$-optimization and managed to connect the performance analysis of $\ell_1$-optimization to the concepts of polytope's neighborliness. The main recognition went along the following lines:
1) Let $C_p^n$ be the regular $n$-dimensional cross-polytope and let the $AC_p^n$ be the polytope one obtains after projecting $C_p^n$ by
$A$; 2) Then the solution of (\ref{eq:l1}) will be exactly the $k$-sparse solution of
(\ref{eq:system}) if and only if polytope $AC_p^n$ is centrally $k$-neighborly (more on the definitions, importance, and many incredible properties of neighborliness can be found in e.g. \cite{DonohoPol,Grunbaum03}). Here we just briefly recall on the basic definitions of neighborliness and central-neighborliness from \cite{DonohoPol}. Namely, a polytope is $k$-neighborly if its every $k+1$ vertices span its a $k$ dimensional face. On the other hand a polytope is centrally $k$-neighborly if its every $k+1$ vertices that do not include any antipodal pair span its a $k$ dimensional face.

The above characterization then enables one to replace studying the success of $\ell_1$-optimization in solving an under-determined system by studying the neighborliness of projected cross-polytopes. Of course, a priori, it is not really clear that the latter problem is any easier then the former one. However, it turns out that it has been explored to some extent in the literature on the geometry of random high-dimensional polytopes. Using the ``sum of angles" result from \cite{AS} (which at its core relies on \cite{PMM,Santalo}) it was established in \cite{DonohoPol} that if $A$ is a random ortho-projector $AC_p^n$ will be centrally $k$-neighborly with overwhelming probability if
\begin{equation}
n^{-1}\log(\com\inter(T^{k},T^{m})\exter(F^{m},C_p^n))<0 \label{eq:conddonpol}
\end{equation}
where $\com=2^{m-k}\binom{n-k-1}{m-k}$, $\inter(T^{k},T^{m})$ is the internal angle at face $T^{k}$ of $T^{m}$, $\exter(F^{m},C_p^n)$ is the external angle of $C_p^n$ at any $m$-dimensional face $F^m$, and $T^{k}$ and $T^{m}$ are the standard $k$ and $m$ dimensional simplices, respectively (more on the definitions and meaning of the internal and external angles can be found in e.g. \cite{Grunbaum03}). Donoho then proceeded by establishing that (\ref{eq:conddonpol}) is equivalent to the following inequality related to the sum/difference of the exponents of $\com,\inter$, and $\exter$:
\begin{equation}
\psinet=\psicom-\psiint-\psiext<0 \label{eq:conddonpolexp}
\end{equation}
where
\begin{eqnarray}
\psicom & = & n^{-1}\log(\com)=(\alpha-\beta) \log(2)+(1-\beta)H(\frac{\alpha-\beta}{1-\beta})\nonumber \\
\psiint & = & n^{-1}\log(\inter(T^{k},T^{m})) \nonumber\\
\psiext & = & n^{-1}\log(\exter(F^{m},C_p^n)) \label{eq:conddonpolexp1}
\end{eqnarray}
and $H(p)=-p\log(p)-(1-p)\log(1-p)$ is the standard entropy function and $\log\binom{n}{pn}=e^{nH(p)}$ is the standard approximation of the binomial factor by the entropy function in the limit of $n\rightarrow\infty$. The rest of the Donoho's approach is the analysis of the closed form expressions for $\inter(T^{k},T^{m})$ and $\exter(F^{m},C_p^n)$ obtained/analyzed in various forms in \cite{BorockyHenk,Ruben,Ver}. In the following two subsections we will separately consider results Donoho established for the internal and the external angle exponents. Relying on the insights from \cite{StojnicCSetam09} we will provide neat characterizations of the exponents that will eventually help us establish the equivalence of results from \cite{DonohoPol} and \cite{StojnicCSetam09}.

\subsection{Internal angle}
\label{sec:internal}

Starting from the explicit formulas for internal angles given in \cite{BorockyHenk} Donoho in \cite{DonohoPol} through a saddle-point integral computation established the following procedure for determining the exponent of the internal angle $\psiint$. Let $\gamma=\frac{\beta}{\alpha}$ and for $s\geq 0$
\begin{eqnarray}
\Phi(s) & = & \frac{1}{\sqrt{2\pi}}\int_{s}^{\infty}e^{-\frac{x^2}{2}}dx\nonumber \\
\phi(s) & = & \frac{1}{\sqrt{2\pi}}e^{-\frac{s^2}{2}}.\label{eq:phis}
\end{eqnarray}
Then one has
\begin{equation}
\psiint(\beta,\alpha)=(\alpha-\beta)\xi_{\gamma}(y_{\gamma})+(\alpha-\beta)\log(2)\label{eq:intang1}
\end{equation}
where
\begin{eqnarray}
y_{\gamma} & = & \frac{\gamma}{1-\gamma}s_{\gamma}\nonumber \\
\xi_{\gamma}(y_{\gamma}) & =  & -\frac{1}{2}y_{\gamma}^2\frac{1-\gamma}{\gamma}-\frac{1}{2}\log(\frac{2}{\pi})+\log(\frac{y_{\gamma}}{\gamma}) \label{eq:intang2}
\end{eqnarray}
and $s_{\gamma}\geq 0$ is the solution of
\begin{equation}
\Phi(s)=(1-\gamma)\frac{\phi(s)}{s}.\label{eq:intang3}
\end{equation}
Now, if one can determine $s_{\gamma}$ then a combination of (\ref{eq:intang1}) and (\ref{eq:intang2}) would give a convenient closed form expression for the exponent $\psiint(\beta,\alpha)$. Finding $s_{\gamma}$ amounts to nothing but solving (\ref{eq:intang3}) over $s$ which for an unknown $\gamma$ could be incredibly hard. At this point we will make a ``bold" guess and say that $t=\frac{1-\alpha}{1-\beta}$ and
\begin{equation}
s_{\gamma}=\sqrt{2}\erfinv(t)=\sqrt{2}\erfinv(\frac{1-\alpha}{1-\beta})\label{eq:sgamma}
\end{equation}
where $\erfinv(\cdot)$ is the inverse of the error function $\erf(\cdot)$ associated with the standard normal random variable ($\erf(r)=\frac{2}{\sqrt{\pi}}\int_0^{r}e^{-q^2}dq$).
Of course, it is rather hard to believe that $s_{\gamma}$ from (\ref{eq:sgamma}) will be the solution of (\ref{eq:intang3}) for every $\gamma=\frac{\beta}{\alpha}$. However, what we hope is that it may be the solution of (\ref{eq:intang3}) for the optimal $\gamma=\frac{\beta_w}{\alpha}$, i.e  for the one for which the net exponent, $\psinet$, in (\ref{eq:conddonpolexp}) is zero (strictly speaking, instead of ``zero" we should say ``smaller than $-\epsilon$ where $\epsilon>0$ is arbitrarily small"; in an effort to make writing and main ideas clearer we will  throughout the rest of the paper almost always ignore $\epsilon$'s). Even a hope like this is fairly out of the blue and it would require an enormous amount of intuition for one to come up with a hopeful guess like the one from (\ref{eq:sgamma}) just by staring at equations (\ref{eq:intang1}-\ref{eq:intang3}) and not knowing the results of \cite{StojnicCSetam09}.

Now what is left to do is to confirm that our guess is actually right. We start by noting that $\Phi(s)=\frac{1}{2}(1-\erf (\frac{s}{\sqrt{2}}))$.
If (\ref{eq:sgamma}) is to be correct then to satisfy (\ref{eq:intang3}) one must have
\begin{equation*}
\frac{1}{2}(1-t)=(1-\gamma)\frac{\phi(s)}{s}=(1-\gamma)\frac{1}{\sqrt{2\pi}}\frac{e^{-(\erfinv (t))^2}}{\sqrt{2}\erfinv (t)}
\end{equation*}
or in a more convenient algebraic form
\begin{equation}
\frac{1-\beta}{\alpha}\sqrt{\frac{2}{\pi}}\frac{e^{-(\erfinv (\frac{1-\alpha}{1-\beta}))^2}}{\sqrt{2}\erfinv (\frac{1-\alpha}{1-\beta})}=1.\label{eq:checkcond1}
\end{equation}
If it eventually turns out that for $\alpha$ and $\beta$ for which (\ref{eq:checkcond1}) holds one also has that $\psinet$ in (\ref{eq:conddonpolexp}) is zero then we could claim that the guess we made for $s_{\gamma}$ in (\ref{eq:sgamma}) is actually correct.

We now proceed with the evaluation of the ``internal exponent" $\psiint$ assuming that both (\ref{eq:sgamma}) and (\ref{eq:checkcond1}) are correct. Plugging (\ref{eq:sgamma}) back in (\ref{eq:intang2}) we obtain
\begin{equation}
y_{\gamma}  = \frac{\gamma}{1-\gamma}s_{\gamma}=\frac{\beta}{\alpha-\beta}\sqrt{2}\erfinv (\frac{1-\alpha}{1-\beta}).\label{eq:evalpsiint1}
\end{equation}
Combining (\ref{eq:intang2}) and (\ref{eq:evalpsiint1}) further we have
\begin{eqnarray}
\xi_{\gamma}(y_{\gamma}) & = &  -\frac{1}{2}y_{\gamma}^2\frac{1-\gamma}{\gamma}-\frac{1}{2}\log(\frac{2}{\pi})+\log(\frac{y_{\gamma}}{\gamma}) \nonumber \\
& = & -\frac{1}{2}\frac{\beta}{\alpha-\beta}(\sqrt{2}\erfinv (\frac{1-\alpha}{1-\beta}))^2-\frac{1}{2}\log(\frac{2}{\pi})+\log(\frac{\alpha}{\alpha-\beta})+
\log(\sqrt{2}\erfinv (\frac{1-\alpha}{1-\beta})).
\label{eq:evalpsiint2}
\end{eqnarray}
Finally plugging $\xi_{\gamma}(y_{\gamma})$ computed in (\ref{eq:evalpsiint2}) back in (\ref{eq:intang1}) we have for the exponent of the internal angle
\begin{multline}
\psiint=-\frac{1}{2}\beta(\sqrt{2}\erfinv (\frac{1-\alpha}{1-\beta}))^2-\frac{\alpha-\beta}{2}\log(\frac{2}{\pi})+(\alpha-\beta)\log(\alpha)\\-(\alpha-\beta)\log(\alpha-\beta)+
(\alpha-\beta)\log(\sqrt{2}\erfinv (\frac{1-\alpha}{1-\beta}))+(\alpha-\beta)\log(2).
\label{eq:evalpsiint3}
\end{multline}

\subsection{External angle}
\label{sec:external}

In this subsection we provide results for the external angle that are analogous to those provided in the previous subsection for the internal angle. In \cite{DonohoPol} Donoho established that the exponent of the external angle can be computed in the following way
\begin{equation}
\psiext(\beta,\alpha)=\min_{y\geq 0} (\alpha y^2 -(1-\alpha)\log(\erf(y))).\label{eq:extang1}
\end{equation}
It was further shown in \cite{DonohoPol} that function $(\alpha y^2 -(1-\alpha)\log(\erf(y)))$ is smooth and convex. If one could solve the above minimization analytically then there would be a neat expression for the exponent of the external angle. As in the previous section, solving this minimization does not appear as an easy task for any fixed $\alpha$ ($\beta$). However, we will again take a ``bold" guess and assume that the solution of the above minimization is
\begin{equation}
y_{ext}=\erfinv(\frac{1-\alpha}{1-\beta}).\label{eq:yextguess}
\end{equation}
It is of course unreasonable to expect that this choice of $y$ would be the solution of the minimization problem in (\ref{eq:extang1}) for every given $\alpha$. However, we do hope that it could be the solution for the optimal pair $(\alpha,\beta)$ (as stated above, the optimal pair $(\alpha,\beta)$ is the one that makes the net exponent $\psinet$ in (\ref{eq:conddonpolexp}) equal to zero). If $y_{ext}$ defined above is to be the solution of the minimization problem in (\ref{eq:extang1}) for the optimal pair $(\alpha,\beta)$ then at the very least one has to have that
\begin{equation}
\frac{d(\alpha y^2 -(1-\alpha)\log(\erf(y)))}{dy} | _{y=y_{ext}}=0.\label{eq:extang2}
\end{equation}
We proceed with checking whether (\ref{eq:extang2}) indeed holds. To that end we have:
\begin{eqnarray}
\frac{d(\alpha y^2 -(1-\alpha)\log(\erf(y)))}{dy} | _{y=y_{ext}} & = & (2\alpha y -\frac{1-\alpha}{\erf(y)}\frac{d\erf(y)}{dy})| _{y=y_{ext}}\nonumber \\
& = & 2\alpha\erfinv(\frac{1-\alpha}{1-\beta})-(1-\beta)\frac{2\sqrt{2}}{\sqrt{2\pi}}e^{-y^2}| _{y=y_{ext}}\nonumber \\
& = & \sqrt{2}\alpha\erfinv(\frac{1-\alpha}{1-\beta})-(1-\beta)\sqrt{\frac{2}{\pi}}e^{-(\erfinv(\frac{1-\alpha}{1-\beta}))^2}\nonumber \\
& = & 0\label{eq:extang3}
\end{eqnarray}
where the last equality follows by our assumption that $(\alpha,\beta)$ are optimal and therefore satisfy (\ref{eq:checkcond1}). Essentially, (\ref{eq:extang3}) shows that if (\ref{eq:checkcond1}) is correct then (\ref{eq:yextguess}) is correct as well.

Combination of (\ref{eq:extang1}) and (\ref{eq:yextguess}) then gives us the following convenient characterization of the ``external exponent" $\psiext$:
\begin{equation}
\psiext=\alpha y_{ext}^2 -(1-\alpha)\log(\erf(y_{ext})))=\alpha(\erfinv(\frac{1-\alpha}{1-\beta}))^2-(1-\alpha)\log(\frac{1-\alpha}{1-\beta}).
\label{eq:evalpsiext1}
\end{equation}

\subsection{Net exponent}
\label{sec:netexp}

In this section we combine the expressions for the ``internal" and ``external" exponents obtained in (\ref{eq:evalpsiint3}) and (\ref{eq:evalpsiext1}), respectively, with the expression for the ``combinatorial" exponent given in (\ref{eq:conddonpolexp1}). Before proceeding further with this exponent combination
we first slightly modify the expression for the combinatorial exponent given in (\ref{eq:conddonpolexp1}).
\begin{eqnarray}
\psicom & = & (\alpha-\beta) \log(2)+(1-\beta)H(\frac{\alpha-\beta}{1-\beta})\nonumber \\
& = & (\alpha-\beta) \log(2)+(1-\beta)(-\frac{\alpha-\beta}{1-\beta}\log(\frac{\alpha-\beta}{1-\beta})-(\frac{1-\alpha}{1-\beta})\log(\frac{1-\alpha}{1-\beta}))\nonumber \\
& = & (\alpha-\beta) \log(2)-(\alpha-\beta)\log(\frac{\alpha-\beta}{1-\beta})-(1-\alpha)\log(\frac{1-\alpha}{1-\beta})
\label{eq:conddonpolexp2}
\end{eqnarray}
Plugging the results from (\ref{eq:evalpsiint3}), (\ref{eq:evalpsiext1}), and (\ref{eq:conddonpolexp2}) back in (\ref{eq:conddonpolexp}) one has
\begin{eqnarray}
\psinet & = & \psicom-\psiint-\psiext\nonumber \\
& = &  (\alpha-\beta) \log(2)-(\alpha-\beta)\log(\frac{\alpha-\beta}{1-\beta})-(1-\alpha)\log(\frac{1-\alpha}{1-\beta})\nonumber \\
& - &
(-\frac{1}{2}\beta(\sqrt{2}\erfinv (\frac{1-\alpha}{1-\beta}))^2-\frac{\alpha-\beta}{2}\log(\frac{2}{\pi})\nonumber \\ & + &(\alpha-\beta)\log(\alpha)- (\alpha-\beta)\log(\alpha-\beta)+
(\alpha-\beta)\log(\sqrt{2}\erfinv (\frac{1-\alpha}{1-\beta}))+(\alpha-\beta)\log(2))\nonumber \\
& - & (\alpha(\erfinv(\frac{1-\alpha}{1-\beta}))^2-(1-\alpha)\log(\frac{1-\alpha}{1-\beta})).
\end{eqnarray}
After canceling all terms that can be canceled one finally has
\begin{eqnarray}
\psinet & = & (\alpha-\beta)\log(\frac{1-\beta}{\alpha})+\frac{\alpha-\beta}{2}\log(\frac{2}{\pi})-
(\alpha-\beta)\log(\sqrt{2}\erfinv (\frac{1-\alpha}{1-\beta}))-  (\alpha-\beta)(\erfinv(\frac{1-\alpha}{1-\beta}))^2\nonumber \\
& = & (\alpha-\beta) (\log(\frac{1-\beta}{\alpha})+\log(\sqrt{\frac{2}{\pi}})-\log(\sqrt{2}\erfinv (\frac{1-\alpha}{1-\beta}))+  \log(e^{-(\erfinv(\frac{1-\alpha}{1-\beta}))^2})\nonumber \\
& = & (\alpha-\beta) \log (\frac{1-\beta}{\alpha}\sqrt{\frac{2}{\pi}}\frac{e^{-(\erfinv(\frac{1-\alpha}{1-\beta}))^2}}{\sqrt{2}\erfinv (\frac{1-\alpha}{1-\beta})})\nonumber \\
& = & 0\label{eq:finalunsign}
\end{eqnarray}
where the last equality follows by assumption (\ref{eq:checkcond1}). Since we obtained that $\psinet=0$ and we never contradicted assumption (\ref{eq:checkcond1}), the assumption must be correct. To be completely rigorous one should add that if an $\alpha$ is given and $\beta_w$ is such that pair $(\alpha,\beta_w)$ satisfies (\ref{eq:checkcond1}) then for any $\beta<\beta_w$ $AC_p^n$ is centrally $\beta n$-neighborly (i.e. one needs $\beta$ to be \emph{strictly} less than $\beta_w$ because (\ref{eq:conddonpolexp}) asserts that one actually needs $\psinet<0$).

We summarize the results from this section in the following theorem.

\begin{theorem}(Geometry-probability equivalence --- General $\x$)
Let $A$ in (\ref{eq:system}) be an $m\times n$ ortho-projector (or an $m\times n$ matrix
with the null-space uniformly distributed in the Grassmanian).
Let $k,m,n$ be large
and let $\alpha=\frac{m}{n}$ and $\beta_w=\frac{k}{n}$ be constants
independent of $m$ and $n$. Let $\erfinv$ be the inverse of the standard error function associated with zero-mean unit variance Gaussian random variable.  Further,
let $\alpha$ and $\beta_w$ be such that
\begin{equation}
\frac{1-\beta_w}{\alpha}\sqrt{\frac{2}{\pi}}\frac{e^{-(\erfinv(\frac{1-\alpha}{1-\beta_w}))^2}}{\sqrt{2}\erfinv (\frac{1-\alpha}{1-\beta_w})}=1.\label{eq:thmweaktheta}
\end{equation}
Then with overwhelming
probability polytope $AC_p^n$ will be centrally $\beta n$-neighborly for any $\beta<\beta_w$.

Further, let
the unknown $\x$ in (\ref{eq:system}) be $k$-sparse and let the location and signs of nonzero elements of $\x$ be arbitrarily chosen but fixed. Then, as shown in \cite{StojnicCSetam09}, for any $\beta<\beta_w$ one with overwhelming
probability has that the solution of (\ref{eq:l1}) is exactly the $\beta n$-sparse $\x$ in (\ref{eq:system}).\label{thm:thmweakthr}
\end{theorem}
\begin{proof}
Follows from the previous discussion through a combination of (\ref{eq:checkcond1}), (\ref{eq:finalunsign}), and the main results of \cite{StojnicCSetam09}.
\end{proof}

The results for the weak threshold obtained from the above theorem have been already plotted in \cite{StojnicCSetam09} and as it was mentioned in \cite{StojnicCSetam09}, they were in an excellent numerical agreement with the ones obtained in \cite{DonohoPol,DonohoUnsigned}
(for the completeness we present the results again in Figure \ref{fig:weak}). Finally, Theorem \ref{thm:thmweakthr} rigorously establishes that the agreement is not only numerical but also analytical and that the weak thresholds obtained in \cite{DonohoPol} and \cite{StojnicCSetam09} are indeed exactly equal to each other.
\begin{figure}[htb]
\centering
\centerline{\epsfig{figure=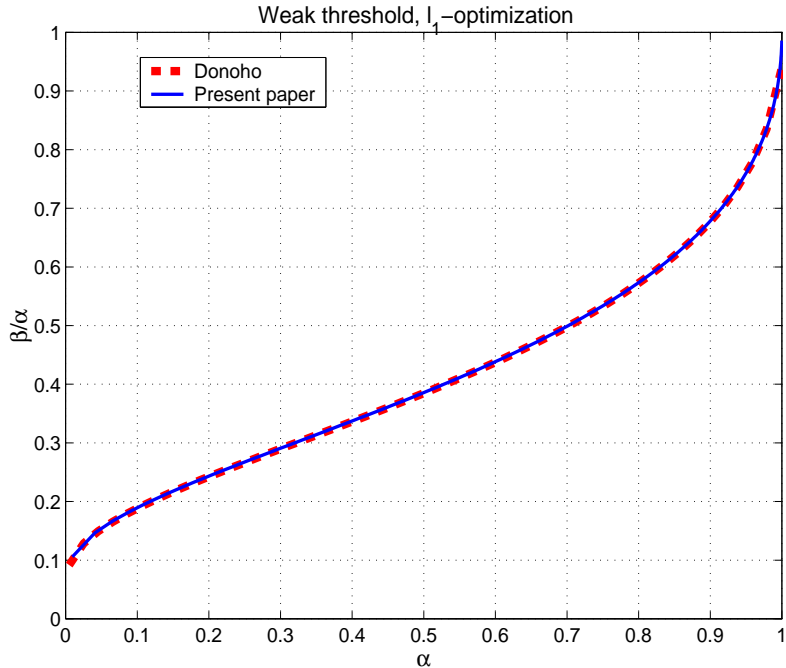,width=10.5cm,height=9cm}}
\caption{\emph{Weak} threshold, $\ell_1$-optimization}
\label{fig:weak}
\end{figure}

\section{Nonnegative $\x$}
\label{sec:signed}

\subsection{Success of $\ell_1$ and neighborliness of projected simplices}
\label{sec:l1neighnon}

In this section we consider a special case of (\ref{eq:system}). We will assume that the nonzero components of $\x$ in (\ref{eq:system}) are all of same sign (say they are all positive). If this is a priori known then instead of using (\ref{eq:l1}) to recover the ``nonnegative" $\x$ in (\ref{eq:system}) one can use (see, e.g. \cite{DonohoSigned,DT,StojnicCSetam09})
\begin{eqnarray}
\mbox{min} & & \|\x\|_{1}\nonumber \\
\mbox{subject to} & & A\x=\y\nonumber \\
& & \x_i\geq 0, 0\leq i\leq n. \label{eq:l1signed}
\end{eqnarray}
Since so to say more structure is imposed on $\x$ (and this structure is made known to the system's solver) one would expect that the recoverable thresholds should be higher in this case than they were in the case of ``general" sparse vectors $\x$. As demonstrated in \cite{DonohoSigned,DT,StojnicCSetam09} the thresholds are indeed higher. Moreover, although the results of \cite{DonohoSigned} and \cite{StojnicCSetam09} were obtained through completely different approaches as demonstrated in \cite{StojnicCSetam09} they happened to be in an excellent numerical agreement. In this section we will rigorously show that the agreement is not only numerical but also algebraic/analytical.

Before proceeding further we quickly recall on and appropriately modify the definition of the weak threshold. The definition of the weak threshold was already introduced in Section \ref{sec:back}. However, that definition was suited for the recovery of general vectors $\x$ considered in the previous section. Here, we slightly modify it so that it fits the scenario of a priori known sign patterns of nonzero elements of $\x$. For any given constant
$\alpha\leq 1$ and \emph{any} given $\x$ with a given fixed location of nonzero components and for which it is known that its nonzero components are say positive
there will be a maximum allowable value of $\beta$ such that
(\ref{eq:l1signed}) finds that given $\x$ in (\ref{eq:system}) with overwhelming
probability. We will refer to this maximum allowable value of
$\beta$ as the nonnegative \emph{weak threshold} and will denote it by $\beta_{w}^{+}$.

The story again starts from Donoho's classic \cite{DonohoUnsigned}. In a follow-up \cite{DonohoSigned} Donoho and Tanner made a key observation that a majority of what was done in \cite{DonohoUnsigned} and was related to the regular $n$-dimensional \emph{cross-polytope} would continue to hold in a slightly modified way if translated to the standard $n$-dimensional \emph{simplex}. In a more mundane language, in \cite{DonohoSigned} Donoho and Tanner took again a geometric approach but this time to the performance analysis of $\ell_1$-optimization from (\ref{eq:l1signed}) and again managed to establish a connection between the performance analysis of (\ref{eq:l1signed}) and the concepts of polytope's neighborliness. This time the main recognition went along slightly different lines:
1) Let $T^n$ be the standard $n$-dimensional simplex and let $AT^n$ be the polytope one obtains after projecting $T^n$ by
$A$; 2) Then the solution of (\ref{eq:l1signed}) will be exactly the \emph{nonnegative} $k$-sparse solution of
(\ref{eq:system}) if and only if polytope $AT^n$ is $k$-neighborly. For completeness we just briefly recall that a polytope is $k$-neighborly if its every $k+1$ vertices span its a $k$ dimensional face.

As in the previous section the above characterization then enables one to replace studying the success of (\ref{eq:l1signed}) in recovering the nonnegative sparse solution of an under-determined system by studying the neighborliness of the projected standard simplex. Of course, as earlier, it is not, a priori, clear that the latter problem is any easier then the former one. However, knowing the results of the previous section (and ultimately of course those of \cite{DonohoUnsigned}) one could now be tempted to believe that polytope type of characterization could be manageable. As discovered in \cite{DT}, it turns out that the neighborliness of randomly projected simplices has been explored to some extent in the literature on the geometry of random high-dimensional polytopes. As in the previous section, using the ``sum of angles" result from \cite{AS,PMM,Santalo} it was established in \cite{DT} that if $A$ is a random ortho-projector $AT^n$ will be $k$-neighborly with overwhelming probability if
\begin{equation}
n^{-1}\log(\comnon\internon(T^{k},T^{m})\externon(F^{m},T^n))<0 \label{eq:conddonpolnon}
\end{equation}
where $\comnon=\binom{n-k-1}{m-k}$, $\internon(T^{k},T^{m})$ is the internal angle at face $T^{k}$ of $T^{m}$, $\externon(T^{m},T^{n-1})$ is the external angle of $T^{n-1}$ at face $T^m$, and $T^{k}$, $T^{m}$, and $T^{n-1}$ are the standard $k$, $m$, and $(n-1)$ dimensional simplices, respectively. The authors in \cite{DT} then proceeded by establishing that (\ref{eq:conddonpolnon}) is equivalent to the following inequality related to the sum/difference of the exponents of $\comnon,\internon$, and $\externon$:
\begin{equation}
\psinetnon=\psicomnon-\psiintnon-\psiextnon<0 \label{eq:conddonpolexpnon}
\end{equation}
where
\begin{eqnarray}
\psicomnon & = & n^{-1}\log(\comnon)=(1-\beta)H(\frac{\alpha-\beta}{1-\beta})\nonumber \\
\psiintnon & = & n^{-1}\log(\internon(T^{k},T^{m})) \nonumber\\
\psiextnon & = & n^{-1}\log(\externon(T^{m},T^{n-1})) \label{eq:conddonpolexp1non}
\end{eqnarray}
and as earlier $H(p)=-p\log(p)-(1-p)\log(1-p)$ is the standard entropy function and $\log\binom{n}{pn}=e^{nH(p)}$ is the standard approximation of the binomial factor by the entropy function in the limit of $n\rightarrow\infty$. The rest of the approach from \cite{DT} is the analysis of the closed form expressions for $\internon(T^{k},T^{m})$ and $\externon(T^{m},T^{n-1})$ obtained/analyzed in various forms in \cite{BorockyHenk,Ruben,Ver}. In the following two subsections we will separately consider results Donoho and Tanner established for the internal and the external angle exponents. Relying on the insights from \cite{StojnicCSetam09} we will provide convenient characterizations of the exponents that will eventually help us establish the rigorous equivalence of results from \cite{DT} and \cite{StojnicCSetam09}.

\subsection{Internal angle --- nonnegative $\x$}
\label{sec:internalnon}

Just by simply looking at formulas (\ref{eq:conddonpolexp}) and (\ref{eq:conddonpolexpnon}) one can hardly see any difference between the definition of the external angle that we in this section and the one that we had in the previous section. The definitions are indeed the sam and the angle's exponents are indeed the same. However, characterizations that we will provide will differ.

To that end we recall on the procedure for determining the exponent of the internal angle $\psiintnon$ (the procedure is of course the same as the one from the previous section and is ultimately the one introduced in \cite{DonohoPol}). As earlier, let $\gamma=\frac{\beta}{\alpha}$ and for $s\geq 0$ let
\begin{eqnarray}
\Phi(s) & = & \frac{1}{\sqrt{2\pi}}\int_{s}^{\infty}e^{-\frac{x^2}{2}}dx\nonumber \\
\phi(s) & = & \frac{1}{\sqrt{2\pi}}e^{-\frac{s^2}{2}}.\label{eq:phisnon}
\end{eqnarray}
Then one has
\begin{equation}
\psiintnon(\beta,\alpha)=(\alpha-\beta)\xi_{\gamma}^{+}(y_{\gamma})+(\alpha-\beta)\log(2)\label{eq:intang1non}
\end{equation}
where
\begin{eqnarray}
y_{\gamma}^{+} & = & \frac{\gamma}{1-\gamma}s_{\gamma}^{+}\nonumber \\
\xi_{\gamma}^{+}(y_{\gamma}^{+}) & =  & -\frac{1}{2}(y_{\gamma}^{+})^2\frac{1-\gamma}{\gamma}-\frac{1}{2}\log(\frac{2}{\pi})+\log(\frac{y_{\gamma}^{+}}{\gamma}) \label{eq:intang2non}
\end{eqnarray}
and $s_{\gamma}^{+}\geq 0$ is the solution of
\begin{equation}
\Phi(s)=(1-\gamma)\frac{\phi(s)}{s}.\label{eq:intang3non}
\end{equation}
As earlier, if one can determine $s_{\gamma}^{+}$ then a combination of (\ref{eq:intang1non}) and (\ref{eq:intang2non}) would give a convenient closed form expression for the exponent $\psiintnon(\beta,\alpha)$. Finding $s_{\gamma}^{+}$ is equivalent to solving (\ref{eq:intang3non}) over $s$ which for a generic $\gamma$ seems to be possibly only numerically. As we have done in the previous section, we will at this point make again a guess and say that $t^{+}=(2\frac{1-\alpha}{1-\beta}-1)$ and
\begin{equation}
s_{\gamma}^{+}=\sqrt{2}\erfinv(t^{+})=\sqrt{2}\erfinv(2\frac{1-\alpha}{1-\beta}-1)\label{eq:sgammanon}
\end{equation}
where as earlier $\erfinv(\cdot)$ is the inverse of the error function $\erf(\cdot)$ associated with the standard normal random variable (i.e $\erfinv(\cdot)$ is the inverse of  $\erf(r)=\frac{2}{\sqrt{\pi}}\int_0^{r}e^{-q^2}dq$).
Of course, as it was the case earlier, $s_{\gamma}^{+}$ from (\ref{eq:sgammanon}) will not be the solution of (\ref{eq:intang3non}) for every $\gamma=\frac{\beta}{\alpha}$. However, we will again hope that it may be the solution of (\ref{eq:intang3non}) for the optimal $\gamma=\frac{\beta_w^{+}}{\alpha}$, i.e  for the one for which the net exponent, $\psinetnon$, in (\ref{eq:conddonpolexpnon}) is zero (again, strictly speaking, instead of ``zero" we should say ``smaller than $-\epsilon$ where $\epsilon>0$ is arbitrarily small"). This hope does not seem any more likely to succeed than the one we made in the previous section unless one is aware of the results of \cite{StojnicCSetam09}.

As it was the case in the previous section, we proceed by trying to confirm that our guess is actually right. We again start by noting that $\Phi(s)=\frac{1}{2}(1-\erf (\frac{s}{\sqrt{2}}))$.
If (\ref{eq:sgammanon}) is to be correct then to satisfy (\ref{eq:intang3non}) one must have
\begin{equation*}
\frac{1}{2}(1-t) =(1-\gamma)\frac{\phi(s)}{s}=(1-\gamma)\frac{1}{\sqrt{2\pi}}\frac{e^{-(\erfinv (t))^2}}{\sqrt{2}\erfinv (t)}
\end{equation*}
or in a more convenient algebraic form
\begin{equation}
\frac{1-\beta}{\alpha}\sqrt{\frac{1}{2\pi}}\frac{e^{-(\erfinv (2\frac{1-\alpha}{1-\beta}-1))^2}}{\sqrt{2}\erfinv (2\frac{1-\alpha}{1-\beta}-1)}=1.\label{eq:checkcond1non}
\end{equation}
Our logic from the previous section still remains in place. Namely, if it eventually turns out that for $\alpha$ and $\beta$ for which (\ref{eq:checkcond1non}) holds one also has that $\psinetnon$ in (\ref{eq:conddonpolexp}) is zero then we could claim that the guess we made for $s_{\gamma}^{+}$ in (\ref{eq:sgammanon}) is actually correct.

We now proceed with the evaluation of the ``internal exponent" $\psiintnon$ assuming that both (\ref{eq:sgammanon}) and (\ref{eq:checkcond1non}) are correct. Plugging (\ref{eq:sgammanon}) back in (\ref{eq:intang2non}) we obtain
\begin{equation}
y_{\gamma}^{+}  = \frac{\gamma}{1-\gamma}s_{\gamma}^{+}=\frac{\beta}{\alpha-\beta}\sqrt{2}\erfinv (2\frac{1-\alpha}{1-\beta}-1).\label{eq:evalpsiint1non}
\end{equation}
Further combination of (\ref{eq:intang2non}) and (\ref{eq:evalpsiint1non}) gives
\begin{eqnarray}
\xi_{\gamma}^{+}(y_{\gamma}^{+}) & = &  -\frac{1}{2}(y_{\gamma}^{+})^2\frac{1-\gamma}{\gamma}-\frac{1}{2}\log(\frac{2}{\pi})+\log(\frac{y_{\gamma}^{+}}{\gamma}) \nonumber \\
& = & -\frac{1}{2}\frac{\beta}{\alpha-\beta}(\sqrt{2}\erfinv (2\frac{1-\alpha}{1-\beta}-1))^2-\frac{1}{2}\log(\frac{2}{\pi})+\log(\frac{\alpha}{\alpha-\beta})+
\log(\sqrt{2}\erfinv (2\frac{1-\alpha}{1-\beta}-1)).\nonumber \\
\label{eq:evalpsiint2non}
\end{eqnarray}
Finally plugging $\xi_{\gamma}^{+}(y_{\gamma}^{+})$ computed in (\ref{eq:evalpsiint2non}) back in (\ref{eq:intang1non}) we have for the exponent of the internal angle
\begin{multline}
\psiintnon=-\frac{1}{2}\beta(\sqrt{2}\erfinv (2\frac{1-\alpha}{1-\beta}-1))^2-\frac{\alpha-\beta}{2}\log(\frac{2}{\pi})+(\alpha-\beta)\log(\alpha)\\-(\alpha-\beta)\log(\alpha-\beta)+
(\alpha-\beta)\log(\sqrt{2}\erfinv (2\frac{1-\alpha}{1-\beta}-1))+(\alpha-\beta)\log(2).
\label{eq:evalpsiint3non}
\end{multline}

\subsection{External angle}
\label{sec:external}

In this subsection we provide the external angle counterparts to results provided in the previous subsection for the internal angle. In \cite{DT} Donoho and Tanner established that the exponent of the external angle can be computed in the following way
\begin{equation}
\psiextnon(\beta,\alpha)=\min_{y\geq 0} (\alpha y^2 -(1-\alpha)\log(\frac{1}{2}(1+\erf(y)))).\label{eq:extang1non}
\end{equation}
It was also shown in \cite{DT} that function $(\frac{1}{2}(1+\erf(y)))$ is smooth and convex. As earlier, solving analytically the above minimization does not appear to be an easy task for a generic fixed $\alpha$ ($\beta$). However, we will again take a guess and assume that the solution of the above minimization is
\begin{equation}
y_{ext}^{+}=\erfinv(2\frac{1-\alpha}{1-\beta}-1).\label{eq:yextguessnon}
\end{equation}
It is again of course unreasonable to expect that this choice of $y$ would be the solution of the minimization problem in (\ref{eq:extang1non}) for every given $\alpha$. However, we do hope that it could be the solution for the optimal pair $(\alpha,\beta)$ (as stated above, the optimal pair $(\alpha,\beta)$ is the one that makes the net exponent $\psinetnon$ in (\ref{eq:conddonpolexpnon}) equal to zero). If $y_{ext}^{+}$ defined above is to be the solution of the minimization problem in (\ref{eq:extang1non}) for the optimal pair $(\alpha,\beta)$ then at the very least one has to have that
\begin{equation}
\frac{d(\alpha y^2 -(1-\alpha)\log(\frac{1}{2}(1+\erf(y))))}{dy} | _{y=y_{ext}^{+}}=0.\label{eq:extang2non}
\end{equation}
To check whether (\ref{eq:extang2non}) holds or not we write:
\begin{multline}
\frac{d(\alpha y^2 -(1-\alpha)\log(\frac{1}{2}(1+\erf(y))))}{dy} | _{y=y_{ext}^{+}}  =  (2\alpha y -\frac{1-\alpha}{\frac{1}{2}(1+\erf(y))}\frac{d\erf(y)}{2dy})| _{y=y_{ext}^{+}} \\
 =  2\alpha\erfinv(2\frac{1-\alpha}{1-\beta}-1)-(1-\beta)\frac{\sqrt{2}}{\sqrt{2\pi}}e^{-y^2}| _{y=y_{ext}^{+}} \\
 =  \sqrt{2}(\sqrt{2}\alpha\erfinv(2\frac{1-\alpha}{1-\beta}-1)-(1-\beta)\sqrt{\frac{1}{2\pi}}e^{-(\erfinv(2\frac{1-\alpha}{1-\beta}-1))^2})=0 \label{eq:extang3non}
\end{multline}
where the last equality follows by our assumption that $(\alpha,\beta)$ are optimal and therefore satisfy (\ref{eq:checkcond1non}). Since, (\ref{eq:extang3non}) shows that (\ref{eq:extang2non}) indeed holds one then has that if (\ref{eq:checkcond1non}) is correct then (\ref{eq:yextguessnon}) is correct as well.

Combination of (\ref{eq:extang1non}) and (\ref{eq:yextguessnon}) then gives us the following convenient characterization of the ``external exponent" $\psiextnon$:
\begin{equation}
\psiextnon=\alpha (y_{ext}^{+})^2 -(1-\alpha)\log(\frac{1}{2}(1+\erf(y_{ext}^{+}))))=\alpha(\erfinv(2\frac{1-\alpha}{1-\beta}-1))^2-(1-\alpha)\log(\frac{1-\alpha}{1-\beta}).
\label{eq:evalpsiext1non}
\end{equation}

\subsection{Net exponent}
\label{sec:netexp}

In this section we combine the expressions for the ``internal" and ``external" exponents obtained in (\ref{eq:evalpsiint3non}) and (\ref{eq:evalpsiext1non}), respectively, with the expression for the ``combinatorial" exponent given in (\ref{eq:conddonpolexp1non}). Before proceeding further
we first slightly modify the expression for the combinatorial exponent given in (\ref{eq:conddonpolexp1non}).
\begin{eqnarray}
\psicomnon & = & (1-\beta)H(\frac{\alpha-\beta}{1-\beta})\nonumber \\
& = & (1-\beta)(-\frac{\alpha-\beta}{1-\beta}\log(\frac{\alpha-\beta}{1-\beta})-(\frac{1-\alpha}{1-\beta})\log(\frac{1-\alpha}{1-\beta}))\nonumber \\
& = & -(\alpha-\beta)\log(\frac{\alpha-\beta}{1-\beta})-(1-\alpha)\log(\frac{1-\alpha}{1-\beta})
\label{eq:conddonpolexp2non}
\end{eqnarray}
Plugging the results from (\ref{eq:evalpsiint3non}), (\ref{eq:evalpsiext1non}), and (\ref{eq:conddonpolexp2non}) back in (\ref{eq:conddonpolexpnon}) one has
\begin{eqnarray}
\psinetnon & = & \psicomnon-\psiintnon-\psiextnon\nonumber \\
& = &  -(\alpha-\beta)\log(\frac{\alpha-\beta}{1-\beta})-(1-\alpha)\log(\frac{1-\alpha}{1-\beta})\nonumber \\
& - &
(-\frac{1}{2}\beta(\sqrt{2}\erfinv (2\frac{1-\alpha}{1-\beta}-1))^2-\frac{\alpha-\beta}{2}\log(\frac{2}{\pi})\nonumber \\ & + &(\alpha-\beta)\log(\alpha)- (\alpha-\beta)\log(\alpha-\beta)+
(\alpha-\beta)\log(\sqrt{2}\erfinv (2\frac{1-\alpha}{1-\beta}-1))+(\alpha-\beta)\log(2))\nonumber \\
& - & (\alpha(\erfinv(2\frac{1-\alpha}{1-\beta}-1))^2-(1-\alpha)\log(\frac{1-\alpha}{1-\beta})).
\end{eqnarray}
After canceling all terms that can be canceled one finally has
\begin{eqnarray}
\psinetnon & = & (\alpha-\beta)(\log(\frac{1-\beta}{\alpha})+\frac{1}{2}\log(\frac{1}{2\pi})-
\log(\sqrt{2}\erfinv (2\frac{1-\alpha}{1-\beta}-1))- (\erfinv(2\frac{1-\alpha}{1-\beta}-1))^2)\nonumber \\
& = & (\alpha-\beta) (\log(\frac{1-\beta}{\alpha})+\log(\sqrt{\frac{1}{2\pi}})-\log(\sqrt{2}\erfinv (2\frac{1-\alpha}{1-\beta}-1))+ \log(e^{-(\erfinv(2\frac{1-\alpha}{1-\beta}-1))^2})\nonumber \\
& = & (\alpha-\beta) \log (\frac{1-\beta}{\alpha}\sqrt{\frac{1}{2\pi}}\frac{e^{-(\erfinv(2\frac{1-\alpha}{1-\beta}-1))^2}}{\sqrt{2}\erfinv (2\frac{1-\alpha}{1-\beta}-1)})\nonumber \\
& = & 0\label{eq:finalsign}
\end{eqnarray}
where the last equality follows by assumption (\ref{eq:checkcond1non}). Since we obtained that $\psinetnon=0$ and we never contradicted assumption (\ref{eq:checkcond1non}), following the logic presented in the previous section, the assumption must be correct. Again, to be completely rigorous one should add that if an $\alpha$ is given and $\beta_w^{+}$ is such that pair $(\alpha,\beta_w^{+})$ satisfies (\ref{eq:checkcond1non}) then for any $\beta<\beta_w^{+}$ $AT^n$ is $\beta n$-neighborly (i.e. one needs $\beta$ to be \emph{strictly} less than $\beta_w^{+}$ because (\ref{eq:conddonpolexpnon}) asserts that one actually needs $\psinetnon<0$).

We summarize the results from this section in the following theorem.

\begin{theorem}(Geometry-probability equivalence --- Nonnegative $\x$)
Let $A$ in (\ref{eq:system}) be an $m\times n$ ortho-projector (or an $m\times n$ matrix
with the null-space uniformly distributed in the Grassmanian).
Let $k,m,n$ be large
and let $\alpha=\frac{m}{n}$ and $\beta_w^{+}=\frac{k}{n}$ be constants
independent of $m$ and $n$. Let $\erfinv$ be the inverse of the standard error function associated with zero-mean unit variance Gaussian random variable.  Further,
let $\alpha$ and $\beta_w^{+}$ be such that
\begin{equation}
\frac{1-\beta_w^{+}}{\alpha}\sqrt{\frac{1}{2\pi}}\frac{e^{-(\erfinv(2\frac{1-\alpha}{1-\beta_w^{+}}-1))^2}}{\sqrt{2}\erfinv (2\frac{1-\alpha}{1-\beta_w^{+}}-1)}=1.
\end{equation}
Then with overwhelming
probability polytope $AT^n$ will be $\beta n$-neighborly for any $\beta<\beta_w^{+}$.

Further, let
the unknown $\x$ in (\ref{eq:system}) be $k$-sparse and \emph{nonnegative} and let the location of nonzero components of $\x$ be arbitrarily chosen but fixed. Then, as shown in \cite{StojnicCSetam09}, for any $\beta<\beta_w^{+}$ one with overwhelming
probability has that the solution of (\ref{eq:l1signed}) is exactly the \emph{nonnegative} $\beta n$-sparse $\x$ in (\ref{eq:system}).\label{thm:thmweakthrnon}
\end{theorem}
\begin{proof}
Follows from the previous discussion through a combination of (\ref{eq:checkcond1non}), (\ref{eq:finalsign}), and the main results of \cite{StojnicCSetam09}.
\end{proof}

The results for the weak threshold obtained from the above theorem have been already plotted in \cite{StojnicCSetam09} and as it was mentioned in \cite{StojnicCSetam09}, they were in an excellent numerical agreement with the ones obtained in \cite{DT,DonohoSigned}
(for the completeness we present the results again in Figure \ref{fig:weaksigned}). Finally, Theorem \ref{thm:thmweakthrnon} rigorously establishes that the agreement is not only numerical but also analytical and that the weak thresholds obtained in \cite{DT} and \cite{StojnicCSetam09} are indeed exactly equal to each other.
\begin{figure}[htb]
\centering
\centerline{\epsfig{figure=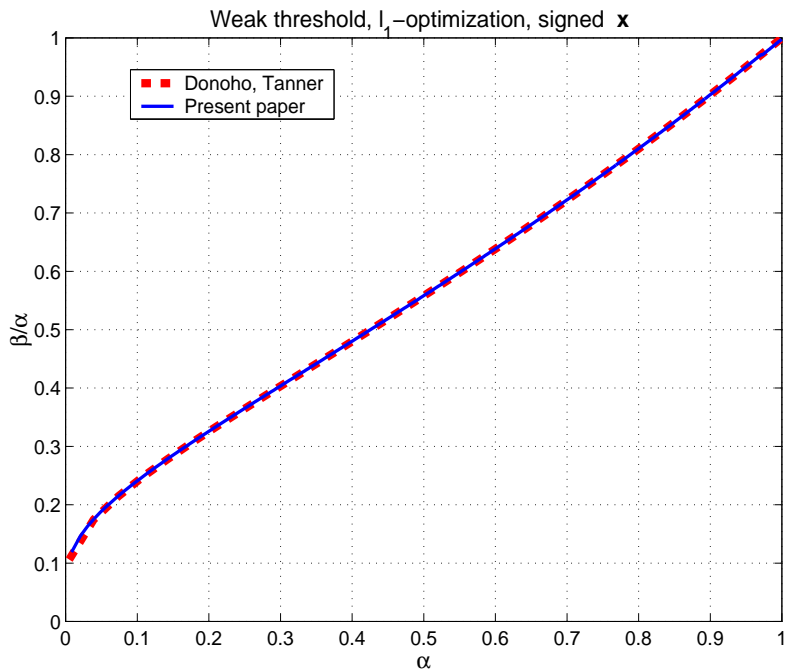,width=10.5cm,height=9cm}}
\vspace{-0.2in} \caption{\emph{Weak} threshold, $\ell_1$-optimization; signed $\x$}
\label{fig:weaksigned}
\end{figure}

\section{Discussion}
\label{sec:discuss}

In this paper we considered under-determined systems of linear equations with sparse solutions.
We focused on solving such systems via a classical polynomial-time
$\ell_1$-optimization algorithm. We also focused on random systems, i.e. on systems where the system matrix is random.

Two different approaches, the geometric one from \cite{DonohoPol} and the probabilistic one from \cite{StojnicCSetam09}, were considered. These approaches were known to provide characterizations of $\ell_1$-optimization success that are in excellent numerical agreement. Here we provided a rigorous proof that the recovery thresholds that one can obtain through one of these approaches are exactly the same as the ones that can be obtained through the other.

We also showed that this remains true when one restricts to under-determined systems with nonnegative and sparse solutions. Namely, we rigorously showed that the nonnegative recovery thresholds one can obtain through either of approaches \cite{DT} and \cite{StojnicCSetam09} are exactly the same as the ones that can be obtained through the other.

An interesting bonus is the following connection with the recent works \cite{DonMalMon09,BayMon10}. Namely, in \cite{DonMalMon09} a belief propagation type of algorithm is put forth as a faster alternative to the standard $\ell_1$-optimization. Its performance was analyzed through a state evolution formalism (which was later made rigorous in \cite{BayMon10}) and the recovery thresholds were computed. Moreover, it was shown in \cite{DonMalMon09} that these thresholds are the same as those computed in \cite{StojnicCSetam09}. The result of this paper confirms that the sparsity recovery abilities of the belief propagation algorithm from \cite{DonMalMon09} are exactly the same not only sa those from \cite{StojnicCSetam09} but also as those from \cite{DonohoPol} and \cite{DT} (an overwhelming numerical evidence of this was of course already presented in \cite{DonMalMon09}).

\begin{singlespace}
\bibliographystyle{plain}
\bibliography{Equivalence}
\end{singlespace}

\end{document}